\documentclass{article}
\usepackage{amsmath, amssymb, amscd, amsthm}
\usepackage{authblk}

\newtheorem{Thm}{Theorem}
\newtheorem{Lem}{Lemma}
\newtheorem{Prop}{Proposition}

\theoremstyle{definition}

\theoremstyle{remark}
\newtheorem{Rem}{Remark}

\title{Propagation of boundary-induced discontinuity in stationary radiative transfer}
\author[1]{I-Kun Chen\thanks{ikun.chen@gmail.com} }
\author[2]{Daisuke Kawagoe\thanks{d.kawagoe@acs.i.kyoto-u.ac.jp}}
\affil[1]{Department of Mathematics, National Taiwan University}
\affil[2]{Graduate School of Informatics, Kyoto University}

\begin{document}

\maketitle

\begin{abstract}
We consider the boundary value problem of the stationary transport equation in the slab domain of general dimensions. In this paper, we discuss the relation between discontinuity of the incoming boundary data and that of the solution to the stationary transport equation. We introduce two conditions posed on the boundary data so that discontinuity of the boundary data propagates along positive characteristic lines as that of the solution to the stationary transport equation. Our analysis does not depend on the celebrated velocity averaging lemma, which is different from previous works. We also introduce an example in two dimensional case which shows that piecewise continuity of the boundary data is not a sufficient condition for the main result.
\end{abstract}

\section{Introduction}

We consider the stationary transport equation:
\begin{align} \label{eq:STE}
&&\xi \cdot \nabla_x f(x, \xi) + \mu_t(x) f(x, \xi) = \mu_s(x) \int_{S^{d-1}} p(x, \xi, \xi^\prime) f(x, \xi^\prime)\,d\sigma_{\xi^\prime}, \nonumber\\
&&(x, \xi) \in \Omega \times S^{d-1},
\end{align}
where $\Omega$ is the slab domain $\mathbb{R}^{d-1} \times (0, 1)$ of the dimension $d \geq 2$ and $S^{d-1}$ is the unit sphere in $\mathbb{R}^d$. The stationary transport equation describes propagation of particles interacting only with the media, for example neutron~\cite{book-fullCase} or photon~\cite{book-fullChan}. The function $f(x, \xi)$ describes the density of particles at the point $x \in \Omega$ with direction $\xi \in S^{d-1}$. The coefficient $\mu_t$ and the product $\mu_s p$ characterize the effect of the media; they are called the attenuation coefficient and the scattering indicatrix, respectively. 

In this paper, we assume that two coefficients $\mu_t$ and $\mu_s$ are nonnegative bounded continuous functions on $\Omega$ satisfying
\begin{equation*}
\inf_{x \in \Omega} \Bigl( \mu_t(x) - \mu_s(x) \Bigr) >0.
\end{equation*}
We note that, from the assumption above, we have 
\begin{equation*}
\underline{\mu_t} := \inf_{x \in \Omega} \mu_t(x) > 0
\end{equation*}
and
\begin{equation*}
M := \sup_{x \in \Omega} \left( \dfrac{\mu_s(x)}{\mu_t(x)} \right) < 1.
\end{equation*}
We also assume that the integral kernel $p$ is a nonnegative bounded continuous function on $\Omega \times S^{d-1} \times S^{d-1}$ which satisfies
\begin{equation*}
\int_{S^{d-1}} p(x, \xi, \xi^\prime)\,d\sigma_{\xi^\prime} = 1
\end{equation*}
for all $(x, \xi) \in \Omega \times S^{d-1}$. We regard the directional derivative $\xi \cdot \nabla_x f(x, \xi)$ as 
\begin{equation*}
\xi \cdot \nabla_x f(x, \xi) := \frac{d}{dt} f(x + t\xi, \xi) |_{t = 0}
\end{equation*}
and the measure $d\sigma_{\xi^\prime}$ is the Lebesgue measure on the sphere $S^{d-1}$.

We pose the incoming boundary condition as follows. Let 
\begin{equation*}
\Gamma_- := \{ (x, \xi) \in \partial \Omega \times S^{d-1} | n(x) \cdot \xi < 0 \},
\end{equation*}
where $n(x)$ is the outer normal vector at $x \in \partial \Omega$. Then, for a given function $f_0$ on $\Gamma_-$, a solution $f$ to the stationary transport equation (\ref{eq:STE}) must satisfy
\begin{equation} \label{eq:BC}
f(x, \xi) = f_0(x, \xi), \quad (x, \xi) \in \Gamma_-.
\end{equation}

In this paper, we discuss the relation between discontinuity of the boundary data and that of the solution to the stationary transport equaton. Aoki et al.~\cite{RefAoki} emphasize importance and significance of this analysis.

We introduce some notations. Let 
\begin{equation*}
X := (\Omega \times S^{d-1}) \cup \Gamma_-.
\end{equation*}
We introduce two functions $\tau_\pm$ on $X$ defined by 
\begin{equation*}
\tau_{\pm}(x, \xi) := \inf \{ t > 0 | x \pm t\xi \not\in \Omega \}.
\end{equation*}
Let
\begin{equation*}
S^{d-1}_{\pm} := \{\xi = (\xi_1, \xi_2, \ldots, \xi_d) \in S^{d-1}| \xi_d \gtrless 0\}
\end{equation*}
and let $\Gamma_{-, \xi}$ and $\Gamma_{-, x}$ be projections of $\Gamma_-$ on $\partial \Omega$ and $S^{d-1}_\pm$ respectively, that is,
\begin{equation*}
\Gamma_{-, \xi} := \{ x \in \partial \Omega | n(x) \cdot \xi < 0 \}, \quad \xi \in S^{d-1}_\pm
\end{equation*}
and
\begin{equation*}
\Gamma_{-, x} := \{ \xi \in S^{d-1} | n(x) \cdot \xi < 0 \}, \quad x \in \partial \Omega.
\end{equation*}
At last, let $disc(f)$ and $disc(f_0)$ be the set of discontinuous points for a function $f$ on $X$ and for a function $f_0$ on $\Gamma_-$, respectively.

The main result in this paper is as follows:
\begin{Thm} \label{thm:MR1}
Suppose that a boundary data $f_0$ is bounded and that it satisfies at least one of the following two conditions:
\begin{enumerate}
\item $f_0(\cdot, \xi)$ is continuous on $\Gamma_{-, \xi}$ for almost all $\xi \in S^{d-1}_\pm$,
\item $f_0(x, \cdot)$ is continuous on $\Gamma_{-, x}$ for almost all $x \in \partial \Omega$.
\end{enumerate}
Then, there exists a unique solution $f$ to the boundary value problem $(\ref{eq:STE})$-$(\ref{eq:BC})$ and the following relation holds:
\begin{equation*}
disc(f) = \{ (x_* + t\xi_*, \xi_*) | (x_*, \xi_*) \in disc(f_0), 0 \leq t < \tau_+(x_*, \xi_*) \}. 
\end{equation*}
\end{Thm}
Here, we call a bounded function $f$ on $X$ a solution to the boundary value problem (\ref{eq:STE})-(\ref{eq:BC}) if it satisfies the stationary transport equation (\ref{eq:STE}) for all $(x, \xi) \in \Omega \times S^{d-1}$ and the boundary condition (\ref{eq:BC}) for all $(x, \xi) \in \Gamma_-$. This theorem means that discontinuity of the boundary data propagates along positive characteristic lines as that of the solution to the stationary transport equation.

\begin{Rem}
If the boundary data $f_0$ is bounded continuous on $\Gamma_-$, then there exists a unique solution $f$. Moreover, since $disc(f_0)$ is the empty set by assumption, $disc(f)$ is also the empty set, which implies that the solution $f$ is also bounded continuous on $X$.
\end{Rem}

Anikonov et al.~\cite{RefAnik} have shown this property assuming condition 1 in a three dimensional bounded convex domain with piecewise continuous coefficients. They also assumed so-called general convexity to these pieces. They made use of this property to solve the inverse problem to determine the coefficient $\mu_t$ from the knowledge of the boundary measurements $f|_{\Gamma_+}$, where 
\begin{equation*}
\Gamma_+ := \{ (x, \xi) \in \partial \Omega \times S^{d-1} | n(x) \cdot \xi > 0 \}
\end{equation*}
and $n(x)$ is again the outer normal vector at $x \in \partial \Omega$. 

On the other hand, Aoki et al.~\cite{RefAoki} have shown the same property assuming condition 2 in the two dimensional half plane domain with $\mu_t = \mu_s = 1$ and $p = 1/2\pi$. We note that this assumption on coefficients violates ours because $\mu_t - \mu_s = 0$ for all $x \in \Omega$. They also assume that the boundary data is periodic with respect to the spacial variable and is independent of the directional variable. We have succeeded in modifying their assumptions. Their analysis is based on the celebrated velocity averaging lemma, which works in $L^2$-framework only when $p$ is constant. On the other hand, $p$ is not constant in our setting. So, we cannot apply their approach directly. However, we overcome this point by $L^\infty$-based discussion.

The ingredient of the rest part in this paper is as follows. In Section 2, we derive integral equations from the boundary value problem, and we show existence and uniqueness of solutions to derived integral equations. In Section 3, we discuss regularity of the solution to integral equations. Especially, we decompose the solution into two parts, the discontinuous part and the continuous part. In Section 4, we check the equivalence between the boundary value problem and integral equations. In other words, we prove that the solution to integral equations indeed satisfies the stationary transport equation under the assumption in Theorem \ref{thm:MR1}. In Section 5, we introduce an example in two dimensional case which shows that piecewise continuity of the boundary data is not a sufficient condition for the main result.

\section{Existence and uniqueness of solutions to the stationary transport equation}
In this section, we derive integral equations from the boundary value problem, and we show existence and uniqueness of solutions to derived integral equations.

For $(x, \xi) \in X$, by integrating the stationary transport equation (\ref{eq:STE}) from $x$ along the negative characteristic line $\{x - t\xi | t > 0\}$ until the line touches the boundary $\partial \Omega$ and by the boundary condition (\ref{eq:BC}) into consideration, we obtain the following integral equations: when $\xi_d \neq 0$,
\begin{align} \label{eq:IE}
f(x, \xi) =& \exp \Bigl(- M_t \bigl(x, \xi; \tau_-\left( x, \xi \right) \bigr) \Bigr) f_0 \bigl( x-\tau_-(x, \xi)\xi, \xi \bigr) \nonumber\\
&+ \int_0^{\tau_-(x, \xi)} \mu_s(x - s\xi) \exp \Bigl(- M_t (x, \xi; s) \Bigr) \nonumber\\
&\quad \times \int_{S^{d-1}} p(x - s\xi, \xi, \xi^\prime) f(x - s\xi, \xi^\prime)\,d\sigma_{\xi^\prime}ds,
\end{align}
and when $\xi_d = 0$,
\begin{align} \label{eq:IE2}
f(x, \xi) =& \int_0^\infty \mu_s(x - s\xi) \exp \Bigl(- M_t (x, \xi; s) \Bigr)\nonumber \\
&\quad \times \int_{S^{d-1}} p(x - s\xi, \xi, \xi^\prime) f(x - s\xi, \xi^\prime)\,d\sigma_{\xi^\prime}ds,
\end{align} 
where
\begin{equation*}
M_t (x, \xi; s) := \int_0^s \mu_t(x - r\xi)\,dr.
\end{equation*}
We call a bounded function $f$ on $X$ a solution to integral equations (\ref{eq:IE})-(\ref{eq:IE2}) if it satisfies integral equations (\ref{eq:IE})-(\ref{eq:IE2}) for all $(x, \xi) \in X$.

We note that, although solutions to the boundary value problem (\ref{eq:STE})-(\ref{eq:BC}) satisfy integral equations (\ref{eq:IE})-(\ref{eq:IE2}), the converse does not hold in general. However, as we will see later in Section 4, under the assumption in Theorem \ref{thm:MR1}, the solution to integral equations (\ref{eq:IE})-(\ref{eq:IE2}) is also the solution to the boundary value problem (\ref{eq:STE})-(\ref{eq:BC}). Therefore, our current task is to find a solution to integral equations (\ref{eq:IE})-(\ref{eq:IE2}).

\begin{Prop}
The solution to integral equations $(\ref{eq:IE})$-$(\ref{eq:IE2})$ is unique, if it exists.
\end{Prop}

\begin{proof}
Let $f_1$ and $f_2$ be two solutions to integral equations (\ref{eq:IE})-(\ref{eq:IE2}). Then the difference $\tilde{f} := f_1 - f_2$ is also bounded on $X$ and satisfies the following integral equation:
\begin{align*}
\tilde{f}(x, \xi) =& \int_0^{\tau_-(x, \xi)} \mu_s(x - s\xi) \exp \Bigl(- M_t (x, \xi; s) \Bigr)\\
& \quad \times \int_{S^{d-1}} p(x, \xi, \xi^\prime) \tilde{f} (x - s\xi, \xi^\prime)\,d\sigma_{\xi^\prime}ds
\end{align*}
for all $(x, \xi) \in X$. Then, we have
\begin{align*}
|\tilde{f}(x, \xi)| \leq& \left( \sup_{(x, \xi) \in X} | \tilde{f} (x, \xi)| \right) \int_0^{\tau_-(x, \xi)} \mu_s(x - s\xi) \exp \Bigl(- M_t (x, \xi; s) \Bigr)ds\\
\leq& \left( \sup_{(x, \xi) \in X} | \tilde{f} (x, \xi)| \right) \int_0^\infty \dfrac{\mu_s(x - s\xi)}{\mu_t(x - s\xi)} \dfrac{d}{ds} \exp \Bigl(- M_t (x, \xi; s) \Bigr)\,ds\\
\leq& M \left( \sup_{(x, \xi) \in X} | \tilde{f} (x, \xi)| \right)
\end{align*}
for all $(x, \xi) \in X$. We emphasize that the supremum in this paper is not the essential supremum, which enables us to justify the pointwise discussion. Therefore,
\begin{equation}
\sup_{(x, \xi) \in X} | \tilde{f} (x, \xi)| \leq M \left( \sup_{(x, \xi) \in X} | \tilde{f} (x, \xi)| \right). \label{ineq:U}
\end{equation}
Since $M < 1$, the inequality (\ref{ineq:U}) implies $\sup_{(x, \xi) \in X} | \tilde{f} (x, \xi)| = 0$, that is, $f_1 = f_2$ for all $(x, \xi) \in X$.
\end{proof}

At last, we prove existence of a solution by iteration. This strategy is standard in the field of radiative transfer. For example, see Anikonov et al.~\cite{RefAnik}. Let us define a family of functions $\{f^{(n)}\}_{n \geq 0}$ on $X$ as follows:
\begin{equation} \label{eq:F0}
f^{(0)}(x, \xi) :=
\begin{cases}
\exp \Bigl(- M_t \bigl(x, \xi; \tau_- (x, \xi) \bigr) \Bigr) f_0(x-\tau_-(x, \xi)\xi, \xi), &\xi_d \neq 0, \\
0, &\xi_d = 0,
\end{cases}
\end{equation}
and
\begin{multline} \label{eq:F1}
f^{(n+1)}(x, \xi) := \int_0^{\tau_-(x, \xi)} \mu_s(x - s\xi) \exp \Bigl(- M_t (x, \xi; s) \Bigr)\\
\times \int_{S^{d-1}} p(x - s\xi, \xi, \xi^\prime) f^{(n)}(x - s\xi, \xi^\prime)\,d\sigma_{\xi^\prime}ds.
\end{multline}

Now we prove that the sum $f := \sum_{n = 0}^\infty f^{(n)}$ is a solution to integral equations (\ref{eq:IE})-(\ref{eq:IE2}). To begin with, we prove that $f := \sum_{n = 0}^\infty f^{(n)}$ is indeed defined on $X$. Especially, we show the following two propositions.  

\begin{Prop} \label{prop:P1}
Suppose that the boundary data $f_0$ is bounded on $\Gamma_-$. Then, $f^{(n)}$ is also bounded on $X$ for all $n \geq 0$.
\end{Prop}

\begin{proof}
We use the induction on $n$. 

For $(x, \xi) \in \Omega \times S^{d-1}$ and $\xi_d = 0$, we have
\begin{equation*}
|f^{(0)}(x, \xi)| = 0 \leq \sup_{(x, \xi) \in \Gamma_-} | f_0 (x, \xi) |,
\end{equation*}
and for otherwise, we have
\begin{align*}
|f^{(0)}(x, \xi)| \leq& \exp \Bigl(- M_t \bigl(x, \xi; \tau_- (x, \xi) \bigr) \Bigr) |f_0(x-\tau_-(x, \xi)\xi, \xi)|\\
\leq& \sup_{(x, \xi) \in \Gamma_-} | f_0 (x, \xi) |.
\end{align*}
These estimates imply that $f^{(0)}$ is bounded on $X$.

Now, we assume that $f^{(n)}$ is bounded on $X$ for some $n \in \mathbb{N}$. Then,
\begin{align}
| f^{(n+1)} (x, \xi) | \leq& \int_0^{\tau_-(x, \xi)}  \mu_s(x - s\xi) \exp \Bigl(- M_t (x, \xi; s) \Bigr) \nonumber \\
&\quad \times \int_{S^{d-1}} p(x - s\xi, \xi, \xi^\prime) | f^{(n)}(x - s\xi, \xi^\prime) | \,d\sigma_{\xi^\prime}ds \nonumber \\
\leq& \left( \sup_{(x, \xi) \in X} | f^{(n)} (x, \xi) | \right) \int_0^{\tau_-(x, \xi)} \mu_s(x - s\xi) \exp \Bigl(- M_t (x, \xi; s) \Bigr)\,ds \nonumber \\
\leq& M \left( \sup_{(x, \xi) \in X} | f^{(n)} (x, \xi) | \right) \label{ineq:E}
\end{align}
for all $(x, \xi) \in X$. This inequality implies that $f^{(n+1)}$ is defined and bounded on $X$. This completes the proof.
\end{proof}

\begin{Prop} \label{prop:P2}
Suppose that the boundary data $f_0$ is bounded on $\Gamma_-$. Then, the sum $\sum_{n = 0}^\infty f^{(n)}(x, \xi)$ converges absolutely and uniformly on $X$.
\end{Prop}

\begin{proof}
From the inequality (\ref{ineq:E}), we have
\begin{align*}
\sup_{(x, \xi) \in X} | f^{(n)} (x, \xi) | \leq& M \left( \sup_{(x, \xi) \in X} | f^{(n-1)} (x, \xi) | \right) \\
\leq& M^n \left( \sup_{(x, \xi) \in X} | f^{(0)}(x, \xi) | \right) \leq M^n \left( \sup_{(x, \xi) \in \Gamma_-} | f_0 (x, \xi) | \right)
\end{align*}
for all $n \geq 0$. Thus, 
\begin{align*}
\sum_{n = 0}^\infty |f^{(n)} (x, \xi)| \leq& \sum_{n = 0}^\infty \ \ \sup_{(x, \xi) \in X} | f^{(n)} (x, \xi) |\\
\leq& \sum_{n = 0}^\infty M^n \left( \sup_{(x, \xi) \in \Gamma_-} | f_0 (x, \xi) |\right) \\
=& \frac{1}{1 - M} \left( \sup_{(x, \xi) \in \Gamma_-} | f_0 (x, \xi) | \right) < \infty,
\end{align*}
which implies absolute and uniform convergence of the sum $\sum_{n = 0}^\infty f^{(n)}(x, \xi)$ on $X$.
\end{proof}

From Proposition \ref{prop:P1} and Proposition \ref{prop:P2}, the sum $f(x, \xi) = \sum_{n = 0}^\infty f^{(n)}(x, \xi)$ converges absolutely and uniformly on $X$ and satisfies
\begin{align*}
f(x, \xi) =& f^{(0)}(x, \xi) + \sum_{n = 0}^\infty f^{(n+1)}(x, \xi)\\
=& f^{(0)}(x, \xi) + \int_0^{\tau_-(x, \xi)} \mu_s(x - s\xi) \exp \Bigl(- M_t (x, \xi; s) \Bigr)\\
& \quad \times \int_{S^{d-1}} p(x - s\xi, \xi, \xi^\prime) \sum_{n = 0}^\infty f^{(n)}(x - s\xi, \xi^\prime)\,d\sigma_{\xi^\prime}ds\\
=& f^{(0)}(x, \xi) + \int_0^{\tau_-(x, \xi)} \mu_s(x - s\xi) \exp \Bigl(- M_t (x, \xi; s) \Bigr)\\
&\quad \times \int_{S^{d-1}} p(x - s\xi, \xi, \xi^\prime) f(x - s\xi, \xi^\prime)\,d\sigma_{\xi^\prime}ds
\end{align*}
for all $(x, \xi) \in X$, which is the pair of integral equations (\ref{eq:IE})-(\ref{eq:IE2}). Thus, the sum  $f(x, \xi) = \sum_{n = 0}^\infty f^{(n)}(x, \xi)$ is the solution to integral equations (\ref{eq:IE})-(\ref{eq:IE2}).

\section{Regularity of the solution}
In this section, we discuss regularity of the solution to integral equations (\ref{eq:IE})-(\ref{eq:IE2}). We decompose the solution $f$ into two parts as below: 
\begin{equation*}
f(x, \xi) = F_0(x, \xi) + F_1(x, \xi),
\end{equation*}
where
\begin{align*}
F_0(x, \xi) &:= f^{(0)}(x, \xi), \\
F_1(x, \xi) &:= \sum_{n = 1}^\infty f^{(n)}(x, \xi).
\end{align*}
From now, we observe discontinuity of $F_0$ and prove continuity of $F_1$. This decomposition is the main idea in our analysis.

\subsection{Discontinuity of $F_0$}
First, we prove the following proposition.
\begin{Prop} \label{prop:F0}
\begin{equation*}
disc(F_0) = \{(x_* + t\xi_*, \xi_*) | (x_*, \xi_*) \in disc(f_0), 0 \leq t < \tau_+(x_*, \xi_*) \}.
\end{equation*}
\end{Prop}

\begin{proof}
Let us recall the explicit formula of $F_0$ (\ref{eq:F0}): when $\xi_d \neq 0$, 
\begin{equation*}
F_0(x, \xi) = \exp \left( -\int_0^{\tau_-(x, \xi)} \mu_t(x - r\xi)\,dr \right) f_0(x-\tau_-(x, \xi)\xi, \xi).
\end{equation*}
$\tau_-$ is continuous on $X$ with $\xi_d \neq 0$ because of its explicit formula:
\begin{equation*}
\tau_-(x, \xi) = 
\begin{cases}
x_d / \xi_d, \quad \xi_d > 0,\\
(x_d - 1) / \xi_d, \quad \xi_d < 0.
\end{cases}
\end{equation*}
Thus, we have when $\xi_d \neq 0$,  
\begin{equation*}
(x, \xi) \in disc(F_0) \Leftrightarrow (x - \tau_-(x, \xi)\xi, \xi) \in disc(f_0),
\end{equation*}
which implies that the statement holds when $\xi_d \neq 0$. 

Thus, only we have to check is continuity of $F_0$ at $(x, \xi) \in \Omega \times S^{d-1}$ with $\xi_d = 0$. In this setting, continuity of $F_0$ with respect to $x$ is obvious from the explicit formula of it (\ref{eq:F0}). So, we focus on continuity of $F_0$ with respect to $\xi$. Since 
\begin{equation*}
\lim_{\xi_d \rightarrow 0} \tau_-(x, \xi) = \infty
\end{equation*}
for all $x \in \Omega$ and since $f_0$ is bounded on $\Gamma_-$, we have 
\begin{align*}
\lim_{\xi_d \rightarrow 0} |F_0(x, \xi)| \leq& \left( \sup_{(x, \xi) \in \Gamma_-} | f_0 (x, \xi) | \right) \lim_{\xi_d \rightarrow 0} \exp \left( -\int_0^{\tau_-(x, \xi)} \mu_t(x - r\xi)\,dr \right)\\
\leq& \left( \sup_{(x, \xi) \in \Gamma_-} | f_0 (x, \xi) | \right) \lim_{\xi_d \rightarrow 0} \exp \left( - \underline{\mu_t} \tau_-(x, \xi) \right) = 0,
\end{align*}
which means $\lim_{\xi_d \rightarrow 0} F_0(x, \xi) = 0$ for all $x \in \Omega$. So $F_0$ is continuous at $(x, \xi) \in \Omega \times S^{d-1}$ with $\xi_d = 0$. This completes the proof.
\end{proof}

\subsection{Continuity of $F_1$}
Secondly, we prove continuity of $F_1$. To do so, we prove by induction that functions $f^{(n)}$, defined above, are bounded continuous on $X$ for all $n \geq 1$. After that, we know from Proposition \ref{prop:P2} that the sum $\sum_{n = 1}^\infty f^{(n)}$ converges uniformly on $X$, which implies that it is also bounded continuous on $X$.

\begin{Lem} \label{lem:f1}
Under the assumption in $Theorem \ref{thm:MR1}$, $f^{(1)}$ is bounded countinuous on $X$.
\end{Lem}
\begin{proof}
Boundedness of $f^{(1)}$ was proved in Section 2, so here we prove continuity of it. By substituting the explicit formula of $f^{(0)}$ (\ref{eq:F0}) for one appeared in the recursive formula (\ref{eq:F1}) with $n = 0$, we have
\begin{align*}
f^{(1)}(x, \xi) =& \int_0^{\tau_-(x, \xi)} \mu_s(x - s\xi) \exp \Bigl(- M_t (x, \xi; s) \Bigr)\\
&\quad \times \int_{S^{d-1}} p(x - s\xi, \xi, \xi^\prime) f^{(0)}(x - s\xi, \xi^\prime)\,d\sigma_{\xi^\prime}ds\\
=& \int_0^{\tau_-(x, \xi)} \mu_s(x - s\xi) \exp \Bigl(- M_t (x, \xi; s) \Bigr) G(x - s\xi, \xi)\,ds,
\end{align*}
where
\begin{equation*}
G(x, \xi) = G_+(x, \xi) + G_-(x, \xi),
\end{equation*}
and
\begin{equation} \label{eq:G}
G_\pm (x, \xi) \!:=\! \int_{S^{d-1}_{\pm}} \! p(x, \xi, \xi^\prime) \! \exp \Bigl(- M_t \bigl(x, \xi^\prime; \tau_- (x, \xi^\prime) \bigr) \Bigr) \! f_0(x - \tau_-(x, \xi^\prime)\xi^\prime, \xi^\prime) \,d\sigma_{\xi^\prime}.
\end{equation}

Then, we introduce the following lemma, whose proof will be appeared later.

\begin{Lem} \label{lem:G1}
Under the assumption in Theorem $\ref{thm:MR1}$, $G$ is bounded continuous on $\Omega \times S^{d-1}$.
\end{Lem}

\begin{Rem}
Lemma \ref{lem:G1} is the very key idea in this paper. Aoki et al.~\cite{RefAoki} proved the same lemma when $p$ is constant, which led them to success. 
\end{Rem}

Admitting Lemma \ref{lem:G1}, we continue to prove Lemma \ref{lem:f1}. Let $\widetilde{G}$ be the zero extension of $G$ to $\mathbb{R}^d \times S^{d-1}$, that is, 
\begin{equation*}
\widetilde{G}(x, \xi) :=
\begin{cases}
G(x, \xi), &\quad (x, \xi) \in \Omega \times S^{d-1},\\
0, &\quad otherwise.
\end{cases}
\end{equation*}
Also let $\widetilde{\mu_t}$ be the zero extension of $\mu_t$ to $\mathbb{R}^d$ and let $\widetilde{M_t}$ be the corresponding $M_t$. Then, $f^{(1)}$ can be written as the following.
\begin{equation*}
f^{(1)}(x, \xi) = \int_0^\infty \widetilde{\mu_s}(x - s\xi) \exp \left(- \widetilde{M_t} (x, \xi; s) \right) \widetilde{G}(x - s\xi, \xi)\,ds,
\end{equation*}
Since the integrand is dominated by 
\begin{equation*}
\left( \sup_{x \in \Omega} \mu_s(x) \right) \left( \sup_{(x, \xi) \in \Omega \times S^{d-1}} | G(x, \xi) | \right) \exp(- \underline{\mu_t} s),
\end{equation*}
which is integrable with respect to $s$ on the half line $\lbrack 0, \infty)$, and since the integrand is continuous at each point $(x, \xi) \in \Omega \times S^{d-1}$ for almost all $s \in \lbrack 0, \infty)$, we apply the dominated convergence theorem to prove continuity of $f^{(1)}$ on $X$. 
\end{proof}

\begin{proof}[Proof of Lemma $\ref{lem:G1}$]
Since boundedness of $G_+$ and $G_-$ is obvious from their formulae (\ref{eq:G}), we focus on discussing continuity of them. At first, we fix a point $(\overline{x}, \overline{\xi}) \in \Omega \times S^{d-1}$ and prove continuity of $G_+$ at the point $(\overline{x}, \overline{\xi})$.

We first suppose that the boundary data $f_0$ satisfies condition 1. Since $\tau_-$ is continuous on $X$, the integrand 
\begin{equation*}
p(x, \xi, \xi^\prime) \exp \Bigl(- M_t \bigl(x, \xi^\prime; \tau_- (x, \xi^\prime) \bigr) \Bigr) f_0(x - \tau_-(x, \xi^\prime)\xi^\prime, \xi^\prime)
\end{equation*}
is continuous at $(\overline{x}, \overline{\xi}) \in \Omega \times S^{d-1}$ for almost all $\xi^\prime \in S^{d-1}_+$. Furthermore, the integrand is bounded by 
\begin{equation*}
\left( \sup_{(x, \xi, \xi^\prime)} p(x, \xi, \xi^\prime) \right) \left( \sup_{(x, \xi) \in \Gamma_-} | f_0 (x, \xi) | \right). 
\end{equation*}
Therefore, we can apply the dominated convergence theorem to conclude that $G_+$ is bounded continuous on $\Omega \times S^{d-1}$.

Next, we suppose that $f_0$ satisfies condition 2. By changing variable of integration $y_0 = (y_1, y_2, \ldots, y_{d-1}, 0) = x - \tau_- (x, \xi^\prime) \xi^\prime$, we have
\begin{align*}
G_+(x, \xi) =& \int_{\mathbb{R}^{d-1}} p \left(x, \xi, \frac{x - y_0}{|x - y_0|} \right) \exp \left(- M_t \left(x, \frac{x - y_0}{|x - y_0|}; |x - y_0| \right) \right)\\
&\quad \times f_0 \left(y_0, \frac{x - y_0}{|x - y_0|} \right) \frac{x_d}{|x - y_0|^d}\,dy_1dy_2 \cdots dy_{d-1},
\end{align*}
where $x_d / |x - y_0|^d$ is the Jacobian of this change.

By condition 2, for almost all $(y_1, y_2, \ldots, y_{d-1}) \in \mathbb{R}^{d-1}$, the integrand
\begin{equation*}
p \left(x, \xi, \frac{x - y_0}{|x - y_0|} \right) \! \exp \left(- M_t \left(x, \frac{x - y_0}{|x - y_0|}; |x - y_0| \right) \right) \! f_0 \left(y_0, \frac{x - y_0}{|x - y_0|} \right) \! \frac{x_d}{|x - y_0|^d}
\end{equation*}
is continuous at $(\overline{x}, \overline{\xi}) \in \Omega \times S^{d-1}$. Furthermore, the integrand is uniformly bounded by 
\begin{equation*}
\left( \left( \sup_{(x, \xi, \xi^\prime)} p(x, \xi, \xi^\prime) \right) \left( \sup_{(x, \xi) \in \Gamma_-} | f_0 (x, \xi) | \right) \right) / (|\tilde{x} - y_0| - \epsilon/2)^d
\end{equation*}
on the neighorhood $B_{\epsilon/2}(\overline{x}) \times S^{d-1}$, where $B_{\epsilon/2}(\overline{x})$ is the closed ball with centre $\overline{x}$ and radius $\epsilon/2$, and $\overline{x}_d > \epsilon >0$. Since the dominant is integrable with respect to $y_1$, $y_2$, $\ldots$, and $y_{d-1}$, we can apply the dominated convergence theorem to conclude that $G_+$ is bounded continuous at $(\overline{x}, \overline{\xi}) \in \Omega \times S^{d-1}$. 

Thus, it follows from the discussion above that $G_+$ is bounded continuous on $\Omega \times S^{d-1}$ if the boundary data $f_0$ satisfies the assumption in Theorem \ref{thm:MR1}. In the same way, we can show that $G_-$ is also bounded continuous on $\Omega \times S^{d-1}$. After all, $G$ itself is also bounded continuous on $\Omega \times S^{d-1}$. 
\end{proof}

\begin{Lem} \label{lem:fn}
Suppose that the function $f^{(n)}$, defined by recursive formulae $(\ref{eq:F0})$-$(\ref{eq:F1})$, is bounded continuous on $X$ for some $n \in \mathbb{N}$. Then, the successive function $f^{(n+1)}$ is also bounded continuous on $X$. 
\end{Lem}

\begin{proof}
As the proof of Lemma \ref{lem:f1}, let $\widetilde{f}^{(n)}$ be the zero extention of $f^{(n)}$ to $\mathbb{R}^d \times S^{d-1}$. Also let $\widetilde{\mu_s}$ and $\widetilde{p}$ be the zero extensions of $\mu_s$ and $p$ to $\mathbb{R}^d$ and $\mathbb{R}^d \times S^{d-1} \times S^{d-1}$, respectively. Then, we have
\begin{multline*}
f^{(n+1)}(x, \xi) = \int_0^\infty \widetilde{\mu_s}(x - s\xi) \exp \left(- \widetilde{M_t} (x, \xi; s) \right)\\
\times \int_{S^{d-1}} \widetilde{p}(x - s\xi, \xi, \xi^\prime) \widetilde{f}^{(n)}(x - s\xi, \xi^\prime)\,d\sigma_{\xi^\prime}ds
\end{multline*}
for all $(x, \xi) \in X$. Since $f^{(n)}$ is continuous on $X$, the integrand 
\begin{equation*}
\widetilde{\mu_s}(x - s\xi) \exp \left(- \widetilde{M_t} (x, \xi; s) \right) \int_{S^{d-1}} \widetilde{p}(x - s\xi, \xi, \xi^\prime) \widetilde{f}^{(n)}(x - s\xi, \xi^\prime)\,d\sigma_{\xi^\prime}
\end{equation*}
is also continuous at each point $(x, \xi) \in X$ for almost all $s \in \lbrack 0, \infty)$. In addition, this integrand is dominated by 
\begin{equation*}
\left( \sup_{x \in \Omega} \mu_s(x) \right) \left( \sup_{(x, \xi) \in X} | f^{(n)} (x, \xi) | \right) \exp \left(- \underline{\mu_t} s \right),
\end{equation*}
which is integrable with respect to $s$. Thus, we can apply Lebesgue's convergence theorem to prove continuity of $f^{(n+1)}$ on $X$.
\end{proof}

By Lemma \ref{lem:f1} and Lemma \ref{lem:fn}, $f^{(n)}$ is bounded countinuous on $X$ for all $n \in \mathbb{N}$, and by Proposition \ref{prop:P2}, the sum $\sum_{n=1}^\infty f^{(n)}(x, \xi)$ converges uniformly on $X$, which implies that the function $F_1$ is also bounded continuous on $X$. 
 
Thus, we succeed to separate the solution into two parts, the discontinuous part $F_0$ and the continuity part $F_1$.

\section{Equivalence between the boundary value problem and derived integral equations}
In this section, we check the equivalence between the boundary value problem (\ref{eq:STE})-(\ref{eq:BC}) and integral equations (\ref{eq:IE})-(\ref{eq:IE2}). As we noted in Section 2, although solutions to the boundary value problem (\ref{eq:STE})-(\ref{eq:BC}) satisfy integral equations (\ref{eq:IE})-(\ref{eq:IE2}), the converse does not hold in general. So, we see that, under the assumption in Theorem 1, the solution to integral equations (\ref{eq:IE})-(\ref{eq:IE2}) is indeed the solution to the boundary value problem (\ref{eq:STE})-(\ref{eq:BC}). In other words, we prove the following two propositions.

\begin{Prop} \label{prop:STE}
Let $f$ be the solution to integral equations $(\ref{eq:IE})$-$(\ref{eq:IE2})$. Suppose that the boundary data $f_0$ satisfies the assumption in $Theorem \ref{thm:MR1}$. Then, the directional derivative $\xi \cdot \nabla_x f(x, \xi)$ is defined for all $(x, \xi) \in \Omega \times S^{d-1}$ and $f$ satisfies the stationary transport equation $(\ref{eq:STE})$.
\end{Prop}

\begin{Prop} \label{prop:BC}
The solution $f$ to integral equations $(\ref{eq:IE})$-$(\ref{eq:IE2})$ satisfies the boundary condition $(\ref{eq:BC})$.
\end{Prop}

\begin{proof}[Proof of Proposition $\ref{prop:STE}$]
Making use of the following equality
\begin{equation*}
\tau_- (x + t\xi, \xi) = \tau_- (x, \xi) + t
\end{equation*}
for $(x, \xi) \in \Omega \times S^{d-1}$ and $t \in \mathbb{R}$ such that $x + t\xi \in \Omega$, we have
\begin{equation*}
f^{(0)}(x + t\xi, \xi) =
\begin{cases}
\exp \Bigl( - M_t \bigl(x + t\xi, \xi; \tau_- (x, \xi) + t \bigr) \Bigr) \! f_0(x - \tau_-(x, \xi)\xi, \xi), \!\!\!\! &\xi_d \neq 0, \\
0, &\xi_d = 0.
\end{cases}
\end{equation*}
When $\xi_d \neq 0$, since 
\begin{equation*}
M_t \bigl(x + t\xi, \xi; \tau_- (x, \xi) + t \bigr) = M_t \bigl(x, \xi; \tau_- (x, \xi) \bigr) - M_t \bigl(x, \xi; -t \bigr),
\end{equation*}
we have
\begin{align*}
\xi \cdot \nabla_x f^{(0)}(x, \xi) \! =& \frac{dM_t}{dt}(x, \xi; -t)|_{t=0} \exp \Bigl(- M_t \bigl(x, \xi; \tau_- (x, \xi) \bigr) \Bigr) \! f_0(x - \tau_-(x, \xi)\xi, \xi)\\
=& -\mu_t(x) \exp \Bigl(- M_t \bigl(x, \xi; \tau_- (x, \xi) \bigr) \Bigr) f_0(x - \tau_-(x, \xi)\xi, \xi)\\
=& -\mu_t(x) f^{(0)}(x, \xi).
\end{align*}
When $\xi_d = 0$, it is obvious that $\xi \cdot \nabla_x f^{(0)}(x, \xi) = 0 = -\mu_t(x) f^{(0)}(x, \xi)$.

Thus, in both cases, we have 
\begin{equation} \label{eq:df0}
\xi \cdot \nabla_x f^{(0)}(x, \xi) = -\mu_t(x) f^{(0)}(x, \xi)
\end{equation}
for all $(x, \xi) \in \Omega \times S^{d-1}$.

Since not only $f^{(1)}$ but also $G$, appeared in the proof of Lemma \ref{lem:f1}, is bounded continuous on $\Omega \times S^{d-1}$, the function
\begin{align*}
f^{(1)}(x + t\xi, \xi) =& \int_0^{\tau_-(x + t\xi, \xi)} \mu_s(x + t\xi - s\xi) \\
&\quad \times \exp \Bigl(- M_t (x + t\xi, \xi; s) \Bigr) G(x + t\xi - s\xi, \xi)\,ds\\
=& \int_{-t}^{\tau_-(x, \xi)} \mu_s(x - s\xi)\\
&\quad \times \exp \Bigl( M_t (x, \xi; -t) - M_t (x, \xi; s) \Bigr) G(x - s\xi, \xi)\,ds
\end{align*} 
is differentiable with respect to $t$ at $t = 0$ for all $(x, \xi) \in \Omega \times S^{d-1}$ and 
\begin{align} \label{eq:df1}
\xi \cdot \nabla_x f^{(1)}(x, \xi) =& \mu_s(x, \xi) G(x, \xi) \nonumber\\
& -\mu_t(x)  \int_0^{\tau_-(x, \xi)} \mu_s(x - s\xi) \exp \Bigl( - M_t (x, \xi; s) \Bigr) G(x - s\xi, \xi)\,ds \nonumber\\
=& \mu_s(x) \int_{S^{d-1}} p(x, \xi, \xi^\prime) f^{(0)}(x, \xi^\prime)\,d\sigma_{\xi^\prime} -\mu_t(x) f^{(1)}(x, \xi) 
\end{align}
for all $(x, \xi) \in \Omega \times S^{d-1}$.

Since the functions $f^{(n)}$ are bouded continuous on $X$ for all $n \in \mathbb{N}$ by Lemma \ref{lem:fn}, the following relation holds from the direct calculation:
\begin{equation} \label{eq:df2}
\xi \cdot \nabla_x f^{(n+1)}(x, \xi) = \mu_s(x) \int_{S^{d-1}} p(x, \xi, \xi^\prime) f^{(n)}(x, \xi^\prime)\,d\sigma_{\xi^\prime} -\mu_t(x) f^{(n+1)}(x, \xi)
\end{equation}
for all $(x, \xi) \in \Omega \times S^{d-1}$ and for all $n \in \mathbb{N}$. 

By Proposition \ref{prop:P2}, we can sum up right hand sides of (\ref{eq:df0}), (\ref{eq:df1}), and (\ref{eq:df2}) to obtain 
\begin{align*}
&\xi \cdot \nabla_x f(x, \xi) = \xi \cdot \nabla_x \sum_{n = 0}^\infty f^{(n)}(x, \xi) = \sum_{n = 0}^\infty \xi \cdot \nabla_x f^{(n)}(x, \xi)\\
=& \sum_{n = 0}^\infty \mu_s(x) \int_{S^{d-1}} p(x, \xi, \xi^\prime) f^{(n)}(x, \xi^\prime)\,d\sigma_{\xi^\prime} - \sum_{n = 0}^\infty \mu_t(x) f^{(n)}(x, \xi)\\
=& \mu_s(x) \int_{S^{d-1}} p(x, \xi, \xi^\prime)  \sum_{n = 0}^\infty f^{(n)}(x, \xi^\prime)\,d\sigma_{\xi^\prime} - \mu_t(x) \sum_{n = 0}^\infty f^{(n)}(x, \xi)\\
=& \mu_s(x) \int_{S^{d-1}} p(x, \xi, \xi^\prime) f (x, \xi^\prime)\,d\sigma_{\xi^\prime} - \mu_t(x) f (x, \xi).
\end{align*}
for all $(x, \xi) \in \Omega \times S^{d-1}$, which is the stationary transport equation (\ref{eq:STE}) itself. Thus, the directional derivative $\xi \cdot \nabla_x f(x, \xi)$ is defined for all $(x, \xi) \in \Omega \times S^{d-1}$ by termwise diffrentiation and the original function $f$ satisfies the stationary transport equation (\ref{eq:STE}).
\end{proof}

\begin{proof}[Proof of Proposition $\ref{prop:BC}$]
For all $(x, \xi) \in \Gamma_-$, 
\begin{equation*}
f^{(n)}(x, \xi) = 
\begin{cases}
f_0(x, \xi), \quad n = 0, \\
0, \quad n \geq 1.
\end{cases}
\end{equation*}
Therefore,
\begin{equation*}
f(x, \xi) = \sum_{n=0}^\infty f^{(n)}(x, \xi) = f_0(x, \xi)
\end{equation*}
for all $(x, \xi) \in \Gamma_-$.
\end{proof}

From proposition 5 and Proposition 6, it follows that the boundary value problem  (\ref{eq:STE})-(\ref{eq:BC}) and integral equations (\ref{eq:IE})-(\ref{eq:IE2}) are equivalent in this setting.

\section{Example for nonequivalence between the stationary transport equation and derived integral equations}
In this section, we introduce an example in two dimensional case which shows that piecewise continuity of the boundary data is not a sufficient condition for the main result. Let $d = 2$ and fix $\overline{x} = (\overline{x_1}, \overline{x_2}) \in \Omega$. We introduce the polar coordinate to $S^1$: 
\begin{equation*}
\xi (\theta) = (\cos \theta, \sin \theta), \quad \theta \in \lbrack 0, 2\pi).
\end{equation*}
We note that, by this coordinate, $S^1_+$ and $S^1_-$ are identified with intervals $(0, \pi)$ and $(\pi, 2\pi)$, respectively. We introduce pieces of $\Gamma_-$ by
\begin{align*}
\Gamma_{-, 1} &:= \{(x, \xi (\theta)) \in \Gamma_- | x_2 = 0, x_1 \geq \overline{x_1} - \overline{x_2} \cot \theta, \theta \in (0, \pi)\}, \\
\Gamma_{-, 2} &:= \{(x, \xi (\theta)) \in \Gamma_- | x_2 = 0, x_1 < \overline{x_1} - \overline{x_2} \cot \theta, \theta \in (0, \pi)\}, \\
\Gamma_{-, 3} &:= \{(x, \xi (\theta) ) \in \Gamma_- | x_2 = 1, \theta \in (\pi, 2\pi)\}.
\end{align*}
We note that $\Gamma_- = \Gamma_{-, 1} \cup \Gamma_{-, 2} \cup \Gamma_{-, 3}$ and
$\Gamma_{-, 1} \cap \Gamma_{-, 2} = \Gamma_{-, 2} \cap \Gamma_{-, 3} = \Gamma_{-, 3} \cap \Gamma_{-, 1} = \emptyset$. 

We take the boundary data $f_0$ as follows:
\begin{equation*}
f_0 (x, \xi) := 
\begin{cases}
1, \quad (x, \xi) \in \Gamma_{-, 1}, \\
0, \quad (x, \xi) \in \Gamma_{-, 2} \cup \Gamma_{-, 3}.
\end{cases}
\end{equation*}
The boundary data $f_0$ is obviously bounded and constant on each $\Gamma_{-, i}$, $i = 1, 2, 3$, which implies that the boundary data $f_0$ is indeed piecewise continuous. 

With this boundary data $f_0$, we define a family of functions $\{ f^{(n)} \}_{n \geq 0}$ on $X$ by recursive formulae (\ref{eq:F0})-(\ref{eq:F1}). Through the same discussion in Section 2, we see that the sum $f = \sum_{n=0}^\infty f^{(n)}$ is still the unique solution to the integral equations (\ref{eq:IE})-(\ref{eq:IE2}). However, the directional derivative $\xi \cdot \nabla_x f(x, \xi)$ of this sum $f$ is not defined at $(\overline{x}, \xi)$ for all $\xi \in S^1$.

In this setting, the function $G_-$, defined by the formula (\ref{eq:G}) in Section 3, is identically zero, which implies that $G_-$ is continuous in $\Omega \times S^1$, while $G_+$, also defined by the formula (\ref{eq:G}) in Section 3, is discontinuous with respect to $x$ at $(\overline{x}, \xi)$ for all $\xi \in S^1$. Thus, $G = G_+ + G_-$ is also discontinuous with respect to $x$ at $(\overline{x}, \xi)$ for all $\xi \in S^1$. Although $G$ is discontinuous with respect to $x$ at one point, $f^{(1)}$ is continuous on $X$. This implies that the directional derivative $\xi \cdot \nabla_x f^{(n)}(x, \xi)$ is defined for all $(x, \xi) \in \Omega \times S^1$ and for all $n \geq 2$, whereas $\xi \cdot \nabla_x f^{(1)}(x, \xi)$ is not defined at $(\overline{x}, \xi) \in \Omega \times S^1$. Since $\xi \cdot \nabla_x f^{(0)}(x, \xi)$ is defined for all $(x, \xi) \in \Omega \times S^1$, we have to conclude that $\xi \cdot \nabla_x \sum_{n=0}^\infty f^{(n)}(x, \xi)$ is not defined at $(\overline{x}, \xi) \in \Omega \times S^1$, which means that the sum $f(x, \xi) = \sum_{n=0}^\infty f^{(n)}(x, \xi)$ is not a solution to the boundary value problem (\ref{eq:STE})-(\ref{eq:BC}). This conclusion implies that this boundary value problem (\ref{eq:STE})-(\ref{eq:BC}) has no solution.

\section*{Acknowledgements}
The authors would like to thank Kazuo Aoki for suggesting this problem. The second author is supported in part by JSPS KAKENHI grant number 15K17572.

\end{document}